\documentclass[conference]{IEEEtran}

\usepackage{amsmath}
\usepackage{amsthm}
\usepackage{amsfonts}
\usepackage{graphicx}
\usepackage[numbers, square, comma, sort&compress]{natbib}
\usepackage{xcolor}
\usepackage{acronym}
\usepackage[bookmarks=false]{hyperref}
\usepackage{caption}
\usepackage{comment}
\usepackage{csquotes}
\usepackage{tikz}
\usetikzlibrary{arrows,shapes.misc,chains,scopes}
\usepackage{pgfplotstable}
\pgfplotsset{compat=1.15}

\newcommand{\noiseguess}[1]{z^{n,{#1}}}
\newcommand{\rvnoise}{N^n}
\newcommand{\noise}{z^n}
\newcommand{\rvchanout}{R^n}

\newcommand{\chanout}{r^n}
\newcommand{\rvdemodout}{Y^n}
\newcommand{\demodout}{y^n}
\newcommand{\codebook}{\mathcal{C}}
\newcommand{\codeword}{c^n}
\newcommand{\cleaned}{\hat{x}^n}
\newcommand{\rvword}{X^n}

\newcommand{\binaryAlphabetN}{\{0,1\}^n}
\newcommand{\realNumbers}{\mathbb{R}}
\newcommand{\guessworkFun}{G}
\newcommand{\guesswork}{\guessworkFun(\rvnoise)}

\newcommand{\mlcodeword}{c^{n,*}}
\newcommand{\secondMostLikelycodeword}{c^{n,**}}
\newcommand{\decodelist}{\mathcal{L}}

\newcommand{\eBCH}{\text{eBCH}}
\newcommand{\RLC}{\text{RLC}}
\newcommand{\CRC}{\text{CRC}}

\acrodef{grand}[GRAND]{Guessing Random Additive Noise Decoding}
\acrodef{llr}[LLR]{log-likelihood ratio}
\acrodef{ldpc}[LDPC]{low-density parity check}
\acrodef{rlc}[RLC]{random linear code}
\acrodef{bch}[BCH]{Bose-Chaudhuri–Hocquenghem}
\acrodef{orbgrand}[ORBGRAND]{Ordered Reliability Bits GRAND}
\acrodef{awgn}[AWGN]{Additive White Gaussian Noise}
\acrodef{cawgn}[CAWGN]{Complex Additive White Gaussian Noise}
\acrodef{bpsk}[BPSK]{Binary Phase-Shift Keying}
\acrodef{harq}[HARQ]{Hybrid Automatic Repeat Request}
\acrodef{crc}[CRC]{Cyclic Redundancy Check}
\acrodef{snr}[SNR]{Signal-to-Noise Ratio}
\acrodef{uer}[UER]{undetected error rate}
\acrodef{er}[ER]{erasure rate}
\acrodef{bler}[BLER]{block error rate}
\acrodef{cascl}[CA-SCL]{CRC-Assisted Successive Cancellation List}

\def\A{ A }
\def  \P{{P}}
\def\PA{{\P(A)}}
\def\PHitIncorrect{\phi}
\newcommand{\B}[1]{{B_{#1}}}
\newcommand{\PB}[1]{{\P(B_{#1})}}

\newcommand{\W}[1]{{W_{(#1)}}}
\newcommand{\Wq}[2]{{W_{(#1), #2}}}
\newcommand{\Wij}[2]{{W_{(#1)}^{#2}}}
\newcommand{\qij}[2]{{q_{#1}^{#2}}}
\newcommand{\qiLj}[1]{{q_1^{L,\{#1\}}}}

\newtheorem{theorem}{Theorem}
\newtheorem{corollary}{Corollary}

\graphicspath{{./img/}}

\title{Upgrade error detection to prediction with GRAND}
\author{
\IEEEauthorblockN{Kevin Galligan}
\IEEEauthorblockA{\textit{Hamilton Institute}\\
\textit{Maynooth University, Ireland}\\
kevin.galligan.2020@mumail.ie}
\and
\IEEEauthorblockN{Peihong Yuan and Muriel M\'edard}
\IEEEauthorblockA{\textit{Research Laboratory for Electronics} \\
\textit{Massachusetts Institute of Technology}\\
Cambridge, USA \\
$\{$phyuan,medard$\}$@mit.edu}
\and
\IEEEauthorblockN{Ken R. Duffy}
\IEEEauthorblockA{\textit{Dept. of ECE \& Dept. Mathematics} \\
\textit{Northeastern University}\\
Boston, USA \\
k.duffy@northeastern.edu}
}

\date{\today}

\begin{document}
\maketitle

\begin{abstract}
\ac{grand} is a family of hard- and soft-detection error correction decoding algorithms that provide accurate decoding of any moderate redundancy code of any length. Here we establish a method through which any soft-input GRAND algorithm can provide soft output in the form of an accurate a posteriori estimate of the likelihood that a decoding is correct or, in the case of list decoding, the likelihood that the correct decoding is an element of the list. Implementing the method adds negligible additional computation and memory to the existing decoding process. The output permits tuning the balance between undetected errors and block errors for arbitrary moderate redundancy codes including CRCs.
\end{abstract}

\begin{IEEEkeywords}
GRAND, soft input, soft output
\end{IEEEkeywords}

\section{Introduction}
For any channel coding scheme, it would be desirable if error correction decoders could produce soft output in the form of a confidence measure in the correctness of a decoded block. Soft output could be used to to make control decisions such as retransmission requests or to tag blocks as erasures for an erasure-correcting code to rectify \cite{lin_error_2004}. A common method for establishing a binary measure of decoding confidence is to append a \ac{crc} to a transmitted message \cite{hashimoto1997_performance,sauter2023_error} prior to error correction encoding that can be used post-decoding to assess consistency. When the block length is large, the addition of a CRC has a negligible effect on the code's rate. One of the goals of modern communications standards, such as 3GPP 5G \cite{3gpp.38.212}, however, is ultra-reliable low-latency communication (URLLC), which requires the use of short packets \cite{durisi2016_toward}. The addition of a \ac{crc} to short packets has a significant effect on the code's rate, and so alternative solutions to evaluate decoding confidence are a topic of active interest, e.g. \cite{sauter2023_error}.

In seminal work on error exponents, Forney \cite{forney1968_exponential} proposed an approximate computation of the correctness probability of a decoded block. Forney's approach necessitates the use of a list decoder, which significantly restricts its applicability, and we shall show that the approach provides an inaccurate estimate in channels with challenging noise conditions, which is a primary region of interest as the output can be used to trigger or suppress a retransmission request, but its potential utility warranted further investigation, e.g. \cite{Hof2010}. With the recent introduction of CRC-Assisted Polar (CA-Polar) codes to communications standards \cite{3gpp.38.212}, Forney's approximation has received renewed interest \cite{sauter2023_error} as one popular method of decoding CA-Polar codes, \ac{cascl} decoding, generates a list of candidate codewords as part of its execution, e.g. \cite{niu2012crc,tal2015list,balatsoukas2015llr,liang2016hardware}. For convolution or trellis codes, the Viterbi algorithm \cite{forney1973_viterbi} can be modified to produce soft output at the sequence level \cite{raghavan1998_reliability}, which has been used in coding schemes with multiple layers of decoding \cite{hagenauer1989_viterbi} and to inform repeat transmission requests \cite{yamamoto1980_viterbi}\cite{raghavan1998_reliability}. The method we develop can be readily used with any moderate redundancy code, can be evaluated without the need to list decode, and the estimate remains accurate in noisy channel conditions.

\ac{grand} is a recently developed family of code-agnostic decoding algorithms that achieve maximum-likelihood decoding for hard detection \citep{duffy_capacity-achieving_2019,galligan2021,An22} and soft detection \citep{solomon20,Duffy19a,duffy2022_ordered,abbas2021list,Duffy23ORBGRANDAI} channels. \ac{grand} algorithms function by sequentially inverting putative noise effects, ordered from most to least likely according to channel properties and soft information, from received signals. The first codeword yielded by inversion of a noise effect is a maximum-likelihood decoding. Since this procedure does not depend on codebook structure, \ac{grand} can decode any moderate redundancy code. Efficient hardware implementations \cite{riaz2021multicodegrand,Riaz23} and syntheses \cite{condo2021_highperformance, abbas2020grand, abbas2021orbgrand, condo2021fixed} for both hard and  soft-detection settings have been translated into taped out circuits that establish the flexibility and energy efficiency of \ac{grand} decoding strategies.

The soft output measure we develop for \ac{grand} is an extremely accurate estimate of the a posteriori probability that a decoding is correct or, in the case of list decoding, the probability that the correct codeword is in the list. We derive these probabilities for uniform at random codebooks and demonstrate empirically that the resulting formulae continue to provide accurate soft output for structured codebooks. The formulae can be used with any algorithm in the GRAND family so long as soft input is available. Calculating the soft output only requires knowledge of the code's dimensions and that the probability of each noise effect query be accumulated during GRAND's normal operation, so computation of the measure does not increase the decoder's algorithmic complexity or memory requirements. In practical terms, the approach provides accurate soft output for single- or list-decoding of any moderate redundancy code of any length and any structure.

\section{Background}
We first define notation used in the rest of the paper. Let $\codebook$ be a codebook containing $2^k$ binary codewords each of length $n$ bits. Let $\rvword:\Omega \to \codebook$ be a codeword drawn uniformly at random from the codebook and let $\rvnoise : \Omega \to \binaryAlphabetN$ denote the binary noise effect that the channel has on that codeword during transmission; that is, $N^n$ encodes the binary difference between the demodulated received sequence and the transmitted codeword, rather than the potentially continuous channel noise. Then $\rvdemodout = \rvword \oplus \rvnoise$ is the demodulated channel output, with $\oplus$ being the element-wise binary addition operator. Let $\rvchanout:\Omega \to \realNumbers^n$ denote soft channel output.
Lowercase letters represent realizations of random variables, with the exception of $\noise$, which is the realization of $\rvnoise$.

All \ac{grand} algorithms operate by progressing through a series of noise effect guesses $\noiseguess{1}, \noiseguess{2}, \ldots \in \binaryAlphabetN$, whose order is informed by channel statistics, e.g. \cite{An22}, or soft input, e.g. \cite{duffy2022_ordered}, until it finds one, $\noiseguess{q}$, that satisfies $\cleaned_q = \demodout \ominus \noiseguess{q} \in \codebook$, where $\ominus$ inverts the effect of the noise on the channel output. If the guesses are in channel-dependent decreasing order of likelihood, then $\noiseguess{q}$ is a maximum-likelihood estimate of $\rvnoise$ and $\cleaned_q$ is a maximum-likelihood estimate of the transmitted codeword $\rvword$. 
Since this guessing procedure does not depend on codebook structure, \ac{grand} can decode any moderate redundancy code as long as it has a method for checking codebook membership. For a linear block code with an $(n-k)\times n$ parity-check matrix $H$, $\cleaned_q$ is a codeword if $H \cleaned_q = 0^n$ \cite{lin_error_2004}, where $\cleaned_q$ is taken to be a column vector and $0^n$ is the zero vector. To generate a decoding list of size $L$, GRAND continues until $L$ codewords are found \cite{abbas2021list,galligan2023_block}.

Underlying \ac{grand} is a race between two random variables, the number of guesses until the true codeword is identified and the number of guesses until an incorrect codeword is identified. Whichever of these processes finishes first determines whether the decoding identified by \ac{grand} is correct. The guesswork function $\guessworkFun : \binaryAlphabetN \to \{1,\ldots, 2^n\}$, which depends on soft input in the soft detection setting, maps a noise effect sequence to its position in GRAND's guessing order, so that $\guessworkFun(\noiseguess{i})=i$. Thus $\guessworkFun(\rvnoise)$ is a random variable that encodes the number of guesses until the transmitted codeword would be identified. If $\W{i} : \Omega \to \{1,\ldots, 2^n-1\}$ is the number of guesses until the $i$-th incorrect codeword is identified, not accounting for the query that identifies the correct codeword, then \ac{grand} returns a correct decoding whenever $\guesswork \leq \W{1}$ and a list of length $L$ containing the correct codeword whenever $\guesswork \leq \W{L}$. Analysis of the race between these two processes leads to the derivation of the soft output in this paper.

Forney's work on error exponents \cite{forney1968_exponential} resulted in an approximation for probabilistic soft output. Given channel output $\chanout$ and a maximum-likelihood decoding output $\mlcodeword \in \codebook$, the probability that the decoding is correct is
\begin{align*}
P(\rvword=\mlcodeword|\rvchanout=\chanout) = \frac{P(\rvchanout=\chanout|\rvword=\mlcodeword)}{\sum_{\codeword\in\codebook} P(\rvchanout=\chanout|\rvword=\codeword)}.
\end{align*}
Based on this formula, Forney derived an optimal threshold for determining whether a decoding should be marked as an erasure. Computing the sum in the formula is infeasible for codebooks of practical size, so Forney 
suggested that, given the second most likely codeword, $\secondMostLikelycodeword \in \codebook$, the correctness probability be approximated by
\begin{align}
\frac{P(\rvchanout=\chanout|\rvword=\mlcodeword)}{P(\rvchanout=\chanout|\rvword=\mlcodeword) + P(\rvchanout=\chanout|\rvword=\secondMostLikelycodeword)},
\label{formula:forney}
\end{align}
which is necessarily no smaller than $1/2$. More generally, given a decoding list $\decodelist \subseteq \codebook$, the denominator can be replaced by $\sum_{\codeword \in \decodelist} P(\rvchanout=\chanout|\rvword=\codeword)$ resulting in an estimate of the correctness probability that is no smaller than 1/$\lvert \decodelist \rvert$. Having the codewords of highest likelihood in the decoding list will give the most accurate approximation as their likelihoods dominate the sum. One downside of this approach is that it requires a list of codewords, which most decoders do not provide. For this reason, a method has recently been proposed to estimate the likelihood of the second-most likely codeword given the first \cite{freudenberger2021_reduced}. A variety of alternative schemes have also been suggested for making erasure decisions, a summary of which can be found in \cite{hashimoto1999_composite}.

\section{GRAND Soft Output}
Throughout this section, we shall assume that the codebook, $\codebook$, consists of $2^k$ codewords drawn uniformly at random from $\binaryAlphabetN$, although the derivation generalises to higher-order symbols. For GRAND algorithms we first derive exact expressions, followed by readily computable approximations, for the probability that the transmitted codeword is not in a decoding list and, as a corollary, that a single-codeword GRAND output is incorrect. In Section \ref{sec:results} we demonstrate the formulae provide excellent estimates for structured codebooks. 

\begin{theorem}[A posteriori likelihood of an incorrect GRAND list decoding for a uniformly random codebook]
\label{thm:list}
Let $G(N^n)$ be the number of codebook queries until the noise effect sequence $N^n$ is identified.
Let $W_1,\ldots,W_{2^k-1}$ be selected uniformly at random without replacement from $\{1,\ldots,2^n-1\}$ and 
define their rank-ordered version $\W{1}<\cdots<\W{2^k-1}$. With the true noise effect not counted, $\W{i}$ corresponds to the location in the guesswork order of the $i$-th erroneous decoding in a codebook constructed uniformly-at-random. 
Define the partial vectors $\Wij{i}{j} = (\W{i},\ldots,\W{j})$ for each $i\leq j\in\{1,\ldots,2^k-1\}$

Assume that a list of $L\geq 1$ codebooks are identified by a GRAND decoder at query numbers $q_1<\ldots<q_L$. Define the associated partial vectors $\qij{i}{j} = (q_i,\ldots,q_{j})$ for each $i\leq j\in\{1,\ldots,2^k-1\}$, and 
\begin{align}
\qiLj{i} = (q_1,\ldots,q_{i-1},q_{i+1}-1,\ldots,q_L-1),
\label{eq:qomit}
\end{align}
which is the vector $\qij{1}{L}$ but with the entry $q_i$ omitted and one subtracted for all entries from $q_{i+1}$ onwards.
Define
\begin{align*}
\PA = \P(G(N^n)>q_L)\P(\Wij{1}{L}=\qij{1}{L}),
\end{align*}
which is associated with the transmitted codeword not being in the list,
and, for each $i\in\{1,\ldots,L-1\}$,
\begin{align*}
\PB{i} &= \P(G(N^n)=q_i)\P(\Wij{1}{L-1}=\qiLj{i}), 
\end{align*}
which is associated with the transmitted codeword being the $i$-th element of
the list,
and
\begin{align*}
\PB{L} & = \P(G(N^n)=q_L)\P(\Wij{1}{L-1}=\qij{1}{L-1},\W{L}\geq q_L),
\end{align*}
which is associated with the transmitted codeword being the final element of the
list. Then the probability that the correct decoding 
is not in the list is
\begin{align}
\frac{\PA}{\PA + \sum_{i=1}^L \PB{i}}.
\label{eq:app}
\end{align}
\end{theorem}

\begin{proof}
For $q\in\{1,\ldots,2^n\}$, define $\Wq{i}{q} = \W{i} + 1_{\{\W{i}\geq q\}}$,
so that any $\W{i}$ that is greater than or equal to $q$ is incremented by one. Note that $\Wq{i}{G(N^n)}$ encodes the locations of erroneous codewords in the guesswork order of a randomly constructed codebook given the value of $G(N^n)$ and, in particular, $\Wq{i}{G(N^n)}$ corresponds the number of queries until the $i$-th incorrect codeword is found given $G(N^n)$.

We identify the event that the decoding is not in the list as
\begin{align*}
\A = \left\{G(N^n)>q_L,\Wij{1}{L}=\qij{1}{L}\right\}
\end{align*}
and the events where the decoding is the $i$-th element of the list by
\begin{align*}
\B{i} = &\left\{\Wij{1}{i-1}=\qij{1}{i-1}, G(N^n)=q_i, \right. \\
        & \left. \Wij{i}{L-1}+1=\qij{i+1}{L}, \W{L}\geq q_{L} \right\}
\end{align*}
where the final condition is automatically met for $i=\{1,\ldots,L-1\}$ but not for $i=L$.
The conditional probability that a GRAND decoding is not one of the elements in the list given that $L$ elements have been found is
\begin{align}
\P\left(\A\middle|\A\bigcup_{i=1}^L\B{i}\right)
= \left. \P(\A) \middle/ \P\left(\A\bigcup_{i=1}^L\B{i}\right).\right.
\label{eq:LLR}
\end{align}
As all of the $\A$ and $\B{i}$ events are disjoint, 
to compute eq. \eqref{eq:LLR} it suffices to simplify $\P(\A)$ and $\P(\B{i})$ for $i\in\{1,\ldots,L\}$ to
evaluate the a posteriori likelihood that the transmitted codeword is not in the list.

Consider the numerator,
\begin{align*}
    \P\left(\A\right) 
    &= \P(G(N^n)>q_L, \Wij{1}{L}=\qij{1}{L})\\
    &= \P(G(N^n)>q_L) \P(\Wij{1}{L}=\qij{1}{L}),
\end{align*}
where we have used the fact that $G(N^n)$ is independent of $\Wij{1}{L}$ by construction. In considering the denominator, we need only be concerned with the terms $\P(\B{i})$
corresponding to a correct codebook being identified at query $q_i$, for which 
\begin{align*}
\P(\B{i})
= &\P(G(N^n)=q_i, \\
& \Wij{1}{i-1}=\qij{1}{i-1}, \Wij{i}{L-1}+1=\qij{i+1}{L}, \W{L}\geq q_{L})\\
= & P(G(N^n)=q_i, \Wij{1}{L-1}=\qiLj{i},  \W{L}\geq q_{L})\\
= & \P(G(N^n)=q_i) \P(\Wij{1}{L-1}=\qiLj{i},  \W{L}\geq q_{L}),
\end{align*}
where we have used the definition of $\qiLj{i}$ in eq. \eqref{eq:qomit} and the independence. Thus the conditional probability that the correct answer is not found in eq. \eqref{eq:LLR} is given in eq. \eqref{eq:app}.
\end{proof}

Specializing to a list size $L=1$, the formula in eq. \eqref{eq:app} for the a posteriori likelihood that decoding is incorrect can be expressed succinctly, as presented in the following corollary.

\begin{corollary}[A posteriori likelihood of an incorrect GRAND decoding for a uniformly random codebook]
\label{cor:single}
The conditional probability that a GRAND decoding is incorrect given a codeword is identified on the $q$-th query is
\begin{align*}
\frac{\P(G(N^n)>q) \P(\W{1}=q)}{\P(G(N^n)=q)\P(\W{1}\geq q)+\P(G(N^n)>q) \P(\W{1}=q)}.
\end{align*}
where $\W{1}$ is equal in distribution to the minimum of $2^{k}-1$ numbers selected uniformly at random without replacement from $\{1,\ldots,2^n-1\}$.
\end{corollary}

In order to compute the a posteriori probability of an incorrect decoding in Theorem \ref{thm:list}, we need to evaluate or approximate: 1) $\P(G(N^n)=q)$ and $\P(G(N^n)\leq q)$; and 2)  $\P(\Wij{1}{L}=\qij{1}{L})$ and $\P(\Wij{1}{L-1}=\qij{1}{L-1},\W{L}\geq q_L)$. During a GRAND algorithm's execution, the precise evaluation of 1) can be achieved by calculating the likelihood of each noise effect query as it is made, $\P(G(N^n)=q)=\P(N^n=z^{n,q})$, and retaining a running sum, $\P(G(N^n) \le q)=\sum_{j=1}^q \P(N^n=z^{n,j})$. For 2), geometric approximations whose asymptotic precision can be verified using the approach described in \cite{duffy_capacity-achieving_2019}[Theorem 2] can be employed, resulting in the following corollaries for list decoding and single-codeword decoding, respectively. 

\begin{corollary}[Approximate a posteriori likelihood of an incorrect GRAND list decoding for a uniformly random codebook]
If each $\W{i}$ given $\W{i-1}$ is assumed to be geometrically distributed with probability of success $(2^k-1)/(2^n-1)$, eq. \eqref{eq:app} describing the a posteriori probability that list decoding does not contain the transmitted codeword can be approximated as
\begin{align}
\frac{\displaystyle \left(1-\sum_{j=1}^{q_L} \P(N^n=z^{n,j})\right) \left(\frac{2^k-1}{2^n-1}\right)}
{\displaystyle \sum_{i=1}^L\P(N^n=z^{n,q_i}) + \left(1-\sum_{j=1}^{q_L} \P(N^n=z^{n,j})\right) \left(\frac{2^k-1}{2^n-1}\right) }
\label{eq:app_approxlist}
\end{align}
\end{corollary}
\begin{proof}
Define the geometric distribution's probability of success to be $\PHitIncorrect=(2^k-1)/(2^n-1)$.
Under the assumptions of the corollary, we have the formulae
\begin{align*}
\P\left(\Wij{1}{L}=\qij{1}{L}\right) &= \left(1-\PHitIncorrect\right)^{q_L-L} \PHitIncorrect^L,
\end{align*}
for $i\in\{1,\ldots,L-1\}$
\begin{align*}
\P\left(\Wij{1}{L-1}=\qiLj{i}\right) & = \left(1-\PHitIncorrect\right)^{q_L-L} \PHitIncorrect^{L-1},
\end{align*}
and
\begin{align*}
\P\left(\Wij{1}{L-1}=\qij{1}{L-1},\W{L}\geq q_L\right)
& = \left(1-\PHitIncorrect\right)^{q_L-L} \PHitIncorrect^{L-1}.
\end{align*}
Using those expressions, simplifying  eq. \eqref{eq:app} gives eq. \eqref{eq:app_approxlist}.
\end{proof}

To a slightly higher precision, one can use the following approximation, which accounts for eliminated queries and is most succintly expressed for a single-codeword decoding.
\begin{corollary}[Approximate a posteriori likelihood of an incorrect GRAND decoding for a uniformly random codebook]
If $\W{1}$ is assumed to be geometrically distributed with probability of success $(2^k-1)/(2^n-q)$ after $q-1$ failed queries, eq. \eqref{eq:app} describing the a posteriori probability that a decoding found after $q_1$ queries is incorrect can be approximated as
\begin{align}
\frac{\displaystyle \left(1-\sum_{j=1}^{q_1} \P(N^n=z^{n,j})\right) \frac{2^k-1}{2^n-{q_1}}}
{\displaystyle \P(N^n=z^{n,q_1})+\left(1-\sum_{j=1}^{q_1} \P(N^n=z^{n,j})\right) \frac{2^k-1}{2^n-q_1}}.
\label{eq:app_approx}
\end{align}
\end{corollary}
\begin{proof}
Under the conditions of the corollary,
\begin{align*}
\P(\W{1}=q_1) = \prod_{i=1}^{q_1-1}\left(1-\frac{2^k-1}{2^n-i}\right) \frac{2^k-1}{2^n-q_1},
\end{align*}
from which eq. \eqref{eq:app} simplifies to \eqref{eq:app_approx}.
\end{proof}

\section{Performance Evaluation}
\label{sec:results}
{\it  Accuracy of soft output.} Armed with the approximate a posteriori probabilities in eq. \eqref{eq:app_approxlist} and \eqref{eq:app_approx}, we investigate their precision for random and structured codebooks. Fig. \ref{fig:ble-estimate-changing-noise} depicts the accuracy of formula \eqref{eq:app_approx} when used for random linear codes $\RLC(64,56)$. For context, Forney's approximation with a list size $L\in \{2,4\}$ is also shown. Transmissions were simulated using a \ac{awgn} channel with \ac{bpsk} modulation. \ac{orbgrand} \cite{duffy2022_ordered} was used for soft-input decoding, which produced decoding lists of the appropriate size for both soft output methods. 

Fig. \ref{fig:ble-estimate-changing-noise} plots the empirical \ac{bler} given the  predicted block error probability evaluated using eq. \eqref{eq:app_approx}. If the estimate was precise, then the plot would follow the line $x=y$, as the predicted error probability and the \ac{bler} would match. As RLCs are linear, codewords are not exactly distributed uniformly in the guesswork order, but the formula provides an accurate estimate. In contrast, Forney's approximation significantly underestimates the error probability, degrades in noisier channels, and has an estimate of no greater than $1/L$. Moreover, GRAND's prediction has been made having only identified a single potential decoding.

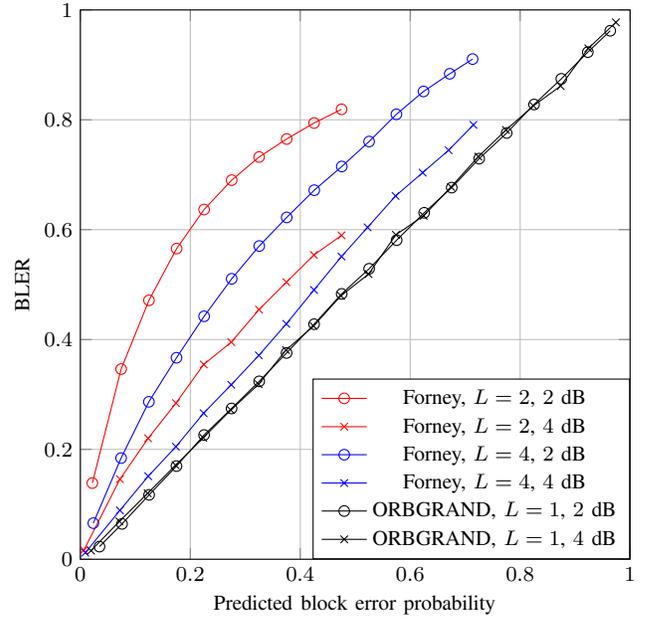
\begin{figure}[t]
\footnotesize
\centering
\begin{tikzpicture}[scale=1]
\begin{axis}[
legend style={at={(1,0)},anchor= south east},
ymin=0,
ymax=1,
width=3.5in,
height=3.5in,
grid=both,
xmin =0,
xmax =1,
xlabel = Predicted block error probability,
ylabel = BLER,
]


\addplot[red, mark=o]
table[]{x y
0.47519 0.81912
0.42522 0.79416
0.37528 0.76504
0.32515 0.73223
0.2752 0.69009
0.2251 0.63663
0.17513 0.56554
0.12483 0.47135
0.0742 0.34618
0.021798 0.13866

};\addlegendentry{Forney, $L=2$, $2~\text{dB}$}

\addplot[red, mark=x]
table[]{x y
0.4751 0.58961
0.42493 0.55391
0.37475 0.50436
0.32495 0.45443
0.27457 0.3953
0.22441 0.35464
0.17397 0.28427
0.12335 0.22014
0.072045 0.14579
0.0063428 0.016489

};\addlegendentry{Forney, $L=2$, $4~\text{dB}$}

\addplot[blue, mark=o]
table[]{x y
0.71349 0.91053
0.67195 0.88375
0.62423 0.85161
0.57492 0.81003
0.52525 0.76041
0.47554 0.71497
0.42547 0.67183
0.37535 0.62226
0.32542 0.57023
0.27526 0.5105
0.22512 0.44225
0.17491 0.36681
0.12464 0.2866
0.074139 0.18418
0.023585 0.065669

};\addlegendentry{Forney, $L=4$, $2~\text{dB}$}

\addplot[blue, mark=x]
table[]{x y
0.71497 0.79082
0.66994 0.74465
0.6226 0.70379
0.57344 0.66135
0.5225 0.60419
0.47506 0.55109
0.42502 0.49022
0.37466 0.42848
0.3247 0.37117
0.27443 0.31767
0.22421 0.26587
0.17385 0.20491
0.12344 0.15111
0.071969 0.088948
0.0088789 0.011883
};\addlegendentry{Forney, $L=4$, $4~\text{dB}$}

\addplot[black, mark=o]
table[]{x y
0.9639 0.96207
0.92316 0.92321
0.87491 0.87459
0.82529 0.82787
0.77545 0.77597
0.72538 0.72926
0.67552 0.67652
0.6254 0.63069
0.57538 0.58086
0.5253 0.5288
0.47536 0.48305
0.42518 0.42789
0.37513 0.37551
0.32523 0.32384
0.27508 0.27446
0.22506 0.22612
0.17489 0.1696
0.12508 0.11723
0.075787 0.064324
0.03476 0.022996

};\addlegendentry{ORBGRAND, $L=1$, $2~\text{dB}$}

\addplot[black, mark=x]
table[]{x y
0.97413 0.97738
0.924 0.93008
0.87367 0.8615
0.82415 0.82538
0.77422 0.78145
0.72395 0.73292
0.67437 0.67773
0.6245 0.62528
0.57404 0.59066
0.52428 0.51912
0.47407 0.47968
0.42424 0.42468
0.37418 0.38076
0.32423 0.31866
0.27369 0.2716
0.22358 0.2211
0.17326 0.17122
0.12247 0.12012
0.071105 0.068971
0.019485 0.015642
};\addlegendentry{ORBGRAND, $L=1$, $4~\text{dB}$}

\end{axis}

\end{tikzpicture}
\vspace{-0.1in}
\caption{The accuracy of soft output when \ac{orbgrand} is used to decode $\RLC(64,57)$. The predicted block error probability is compared to the measured BLER. If the soft output was perfectly accurate, then the data would follow the line $x=y$.}
\vspace{-0.2in}
\label{fig:ble-estimate-changing-noise}
\end{figure}
Fig. \ref{fig:list-estimate-varying-parameters} is similar to Fig. \ref{fig:ble-estimate-changing-noise} except for lists, where a list error occurs when the transmitted codeword is not in the decoding list. The measured list-BLER is plotted against the predicted list-BLER in eq. \eqref{eq:app_approxlist}. The prediction can be seen to be robust to channel condition, list size, and code structure. 

\begin{figure}[t]
\footnotesize
\centering
\begin{tikzpicture}[scale=1]
\begin{axis}[
legend style={at={(1,0)},anchor= south east, font=\scriptsize},
ymin=0,
ymax=1,
width=3.5in,
height=3.5in,
grid=both,
xmin =0,
xmax =1,
xlabel = Predicted list error probability,
ylabel = list-BLER,
]

\addplot[red, mark=o]
table[]{x y
0.96418 0.96197
0.92334 0.92382
0.87503 0.87651
0.82547 0.82863
0.77553 0.78024
0.72554 0.72688
0.6755 0.68423
0.62542 0.63459
0.57543 0.58546
0.5253 0.53848
0.47532 0.48368
0.42513 0.42607
0.37528 0.38719
0.32505 0.3374
0.27504 0.28475
0.22505 0.22811
0.17502 0.17959
0.12508 0.12796
0.075741 0.086516
0.034915 0.04021
};\addlegendentry{RLC, $L=2$, $2~\text{dB}$}

\addplot[blue, mark=o]
table[]{x y
0.96415 0.96185
0.92333 0.92192
0.87503 0.87439
0.82543 0.82623
0.77559 0.77795
0.7255 0.73138
0.67551 0.68424
0.62553 0.63117
0.57536 0.58958
0.52542 0.53711
0.47535 0.48823
0.42519 0.43815
0.37523 0.38502
0.3251 0.33117
0.27496 0.27933
0.22501 0.2353
0.17509 0.18071
0.12495 0.13462
0.075678 0.08896
0.034831 0.041306
};\addlegendentry{RLC, $L=4$, $2~\text{dB}$}

\addplot[red, mark=x]
table[]{x y
0.97378 0.9592
0.92414 0.9247
0.87346 0.87506
0.82358 0.8195
0.77363 0.77682
0.72429 0.72579
0.67456 0.67069
0.62445 0.6211
0.57417 0.58587
0.52417 0.53526
0.47409 0.48312
0.42423 0.42995
0.37398 0.38038
0.32375 0.32527
0.27396 0.28173
0.2236 0.23875
0.17321 0.18269
0.12251 0.12796
0.071125 0.076159
0.019429 0.020504
};\addlegendentry{RLC, $L=2$, $4~\text{dB}$}

\addplot[blue, mark=x]
table[]{x y
0.97475 0.97593
0.92317 0.92182
0.87337 0.87781
0.82414 0.82589
0.774 0.7646
0.72453 0.71852
0.67449 0.67311
0.62415 0.64053
0.57434 0.58159
0.52433 0.53506
0.47434 0.46568
0.4243 0.43567
0.37423 0.38044
0.3241 0.33378
0.2738 0.28676
0.22363 0.22735
0.17331 0.18481
0.12261 0.1338
0.071096 0.08575
0.019401 0.021216

};\addlegendentry{RLC, $L=4$, $4~\text{dB}$}

\addplot[red, mark=o, dashed, mark options=solid]
table[]{x y
0.96098 0.96279
0.91922 0.92543
0.87259 0.8792
0.82387 0.8293
0.77467 0.77888
0.72494 0.72772
0.67505 0.67241
0.62523 0.62199
0.57534 0.56634
0.52528 0.51126
0.4753 0.45617
0.42522 0.40028
0.37518 0.34606
0.32508 0.2945
0.27522 0.24112
0.22506 0.19143
0.17494 0.14313
0.1246 0.099159
0.074203 0.055991
0.02533 0.016624
};\addlegendentry{eBCH, $L=2$, $2~\text{dB}$}

\addplot[blue, mark=o, dashed, mark options=solid]
table[]{x y
0.96094 0.96305
0.91977 0.92441
0.87153 0.87779
0.8233 0.83009
0.77438 0.7783
0.7238 0.72818
0.67498 0.67222
0.62581 0.622
0.57551 0.56588
0.52701 0.51171
0.47572 0.4677
0.42457 0.41912
0.37348 0.35509
0.32627 0.31309
0.27465 0.25977
0.22536 0.20405
0.17488 0.15353
0.12339 0.114376
0.073616 0.060001
0.026226 0.020222
};\addlegendentry{eBCH, $L=4$, $2~\text{dB}$}

\addplot[red, mark=x, dashed, mark options=solid]
table[]{x y
0.9609 0.9638
0.9217 0.9256
0.8720 0.8795
0.8239 0.8301
0.7734 0.7775
0.7254 0.7272
0.6744 0.6740
0.6257 0.6227
0.5761 0.5655
0.5265 0.5099
0.4751 0.4572
0.4240 0.4000
0.3757 0.3464
0.3228 0.2945
0.2754 0.2417
0.2253 0.1938
0.1749 0.1411
0.1236 0.0991
0.0748 0.0543
0.0266 0.0178
};\addlegendentry{eBCH, $L=2$, $4~\text{dB}$}

\addplot[blue, mark=x, dashed, mark options=solid]
table[]{x y
0.96599 0.96435
0.92557 0.92078
0.86104 0.87036
0.82436 0.8329
0.77805 0.77405
0.72876 0.73478
0.66663 0.67442
0.6247 0.6218
0.57134 0.56857
0.52264 0.5164
0.47445 0.46502
0.43985 0.42184
0.37251 0.35432
0.32284 0.31325
0.27474 0.26722
0.22544 0.20456
0.17645 0.15154
0.12012 0.11552
0.070416 0.05589
0.038501 0.0209
};\addlegendentry{eBCH, $L=4$, $4~\text{dB}$}

\addplot[gray,dashed]
table[]{x y
0 0
1 1
};

\end{axis}

\end{tikzpicture}
\vspace{-0.1in}
\caption{The accuracy of the predicted list error probability compared to the measured list-BLER. Parameters as in Fig. \ref{fig:ble-estimate-changing-noise}, but with varying the channel noise,  list size $L \in \{2, 4\}$, and the code types of $\RLC(64,57)$ and $\eBCH(64,57)$.}
\vspace{-0.2in}
\label{fig:list-estimate-varying-parameters}
\end{figure}
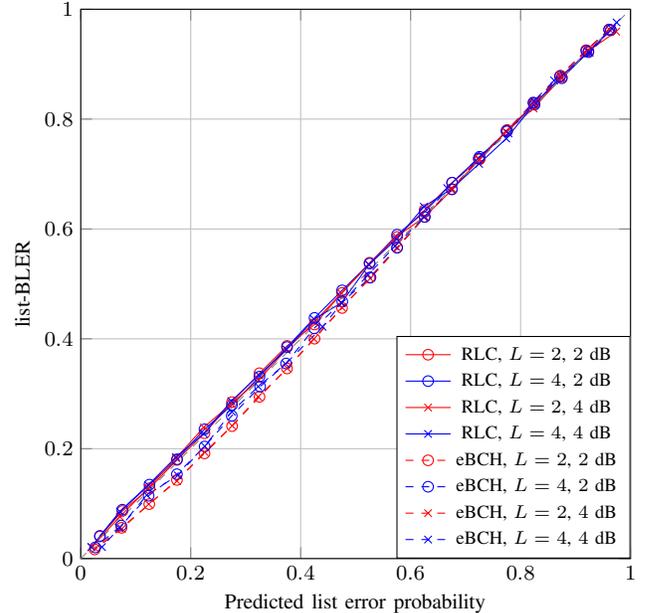

{\it Application to error detection.} A common method used to detect errors is to append a \ac{crc} to a message and declare an erasure at receiver if there is an inconsistency. As GRAND algorithms can decode any code, one use of GRAND soft output is to upgrade CRCs so that they are used for both error correction and error detection, by decoding the block but returning an erasure if the decoding is too unreliable. The \ac{bler} that results from this process is composed of both undetected block errors and erasures.

Fig. \ref{fig:crc-detection} depicts \ac{uer} and \ac{bler} of a $\CRC(64,56)$ code. 
Two methods of error control are compared: 1) the CRC is checked for consistency and an erasure is declared if it fails; 2) \ac{orbgrand} is used to correct errors and an erasure is declared if the estimated error probability is greater than a threshold $\epsilon$. The advantage of 2) is that error correction with GRAND results in significantly reduced \ac{bler} while tailoring the error detection to a target \ac{uer} by modifying the threshold accordingly. 

\begin{figure}[t]
    \footnotesize
    \centering
    \begin{tikzpicture}[scale=1]
\begin{semilogyaxis}[
legend style={at={(1,1)},anchor= north east, font=\scriptsize},
ymin=0,
ymax=1,
width=3.5in,
height=3.5in,
grid=both,
xlabel = Eb/N0 (dB),
ylabel = BLER / UER,
]

\addplot[red, mark=o]
table[]{x y
3.0 0.8637230853093486
3.5 0.7894321927350573
4.0 0.6873841883688295
4.5 0.5737532326080423
5.0 0.45115872407331913951
};\addlegendentry{CRC}

\addplot[black, mark=o]
table[]{x y
3.0 0.85699648600854594349
3.5 0.68074162972397067595
4.0 0.46084419889502764578
4.5 0.24879759137162238414
5.0 0.10776775695108266329
5.5 0.03598699606745639490
};\addlegendentry{$\epsilon$=0.025}

\addplot[brown, mark=o]
table[]{x y
3.0 0.59960012627591285916
3.5 0.37880464858882123558
4.0 0.19140320822301823411
4.5 0.08221053390013689788
5.0 0.02670713201820940769
5.5 0.00724994925035524727
};\addlegendentry{$\epsilon$=0.1}

\addplot[blue, mark=o]
table[]{x y
3.0 0.27265238879736408695
3.5 0.16212003117692908027
4.0 0.07229778095919828229
4.5 0.02548241159634382844
5.0 0.00874569768097951776
5.5 0.00242706703865567912
};\addlegendentry{$\epsilon$=0.5}

\addplot[red, dashed, mark=x, mark options=solid]
table[]{x y
3.0 0.00097340653350465290
3.5 0.00033980087668626185
4.0 0.00019382699777486606
4.5 0.00009571008267436942
};

\addplot[black, mark=x, dashed, mark options=solid]
table[]{x y
3.0 0.00018974035929234437
3.5 0.00030615493889147422
4.0 0.00044198895027624310
4.5 0.00038111208506421741
5.0 0.00025975645235027636
5.5 0.00015078771502328162
};

\addplot[brown, dashed, mark=x, mark options=solid]
table[]{x y
3.0 0.00526149636956750513
3.5 0.00691754288876591031
4.0 0.00778694907335306021
4.5 0.00360256502629872476
5.0 0.00168605631428089689
5.5 0.00072499492503552471
};

\addplot[blue, dashed, mark=x, mark options=solid]
table[]{x y
3.0 0.08237232289950575936
3.5 0.07794232268121589757
4.0 0.03579098067287043994
4.5 0.01624965377158157226
5.0 0.00592450488066354440
5.5 0.00162073247431824413
};

\end{semilogyaxis}

\end{tikzpicture}
    \vspace{-0.1in}
    \caption{The \ac{uer} (dotted lines) and \ac{bler} (solid lines) of a $\CRC(64,56)$ code 
    with two methods of error control: 1) CRC used for error detection; 2) \ac{orbgrand} performs error correction using the CRC, then erasures are declared if the predicted block error probability exceeds $\epsilon$.}
    \vspace{-0.3in}
    \label{fig:crc-detection}
\end{figure}
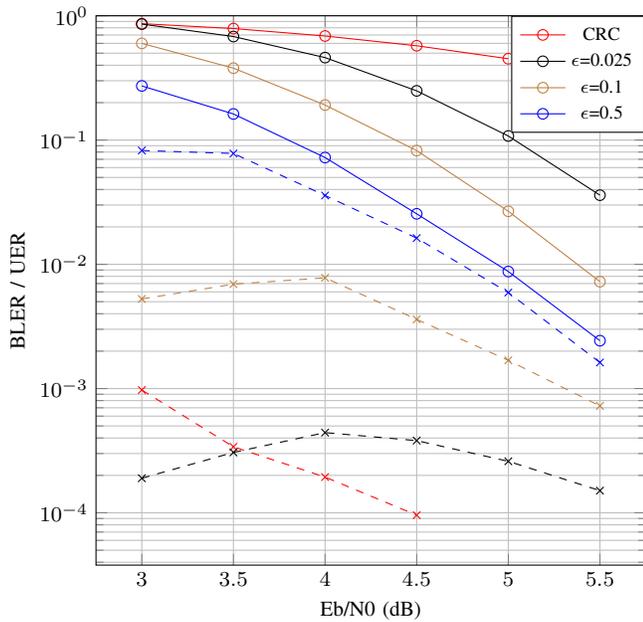

As another example, Fig. \ref{fig:ca-polar-detection} depicts the error detection and correction performance of an $\eBCH(64,51)$ code with \ac{orbgrand} decoding when an error probability threshold $\epsilon$ is used for erasure decisions. Shown for comparison is \ac{cascl} decoding \cite{tal2015_list} of a $(64,51+6)$ 5G polar code \cite{3gpp.38.212} concatenated with the 6-bit CRC \texttt{0x30}, generating a list of $8$ candidates from which the most likely of those whose CRC is consistent is declared to be the decoding. If no element of the CA-SCL list has a CRC that matches, it is treated as an erasure. With an appropriately chosen $\epsilon$, both methods achieve a similar \ac{bler}, but \ac{orbgrand} is shown to achieve a \ac{uer} that is almost an order of magnitude lower than \ac{cascl} in the 2 to 3dB Eb/N0 range. In less noisy conditions, it still achieves a gain of 0.5dB.

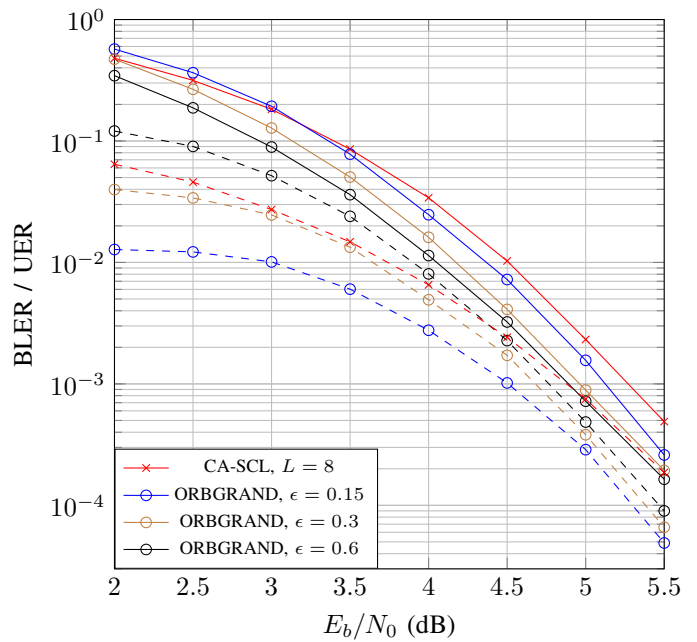
\begin{figure}
    \centering
    \begin{tikzpicture}[scale=1]
\begin{semilogyaxis}[
legend style={at={(0,0)},anchor= south west,font=\scriptsize},
ymin=0.00003,
ymax=1,
width=3.5in,
height=3.5in,
grid=both,
xmin = 2,
xmax = 5.5,
xlabel = $E_b/N_0$ (dB),
ylabel = BLER / UER
]

\addplot[red, mark = x]
table[]{x y
2 0.48021
2.5 0.31561
3 0.18246
3.5 0.085634
4 0.034092
4.5 0.010282
5 0.0023203
5.5 0.00048931
};\addlegendentry{CA-SCL, $L=8$}

\addplot[blue, mark = o]
table[]{x y
2 0.57125
2.5 0.36492
3 0.19309
3.5 0.077915
4 0.024714
4.5 0.00723
5 0.0015633
5.5 0.00026
};\addlegendentry{ORBGRAND, $\epsilon=0.15$}

\addplot[brown, mark = o]
table[]{x y
2 0.47147
2.5 0.26636
3 0.12791
3.5 0.050477
4 0.016106
4.5 0.0041062
5 0.00088874
5.5 0.000194
};\addlegendentry{ORBGRAND, $\epsilon=0.3$}

\addplot[black, mark = o]
table[]{x y
2 0.34494
2.5 0.18747
3 0.089186
3.5 0.036138
4 0.011396
4.5 0.0032351
5 0.00072
5.5 0.000164

};\addlegendentry{ORBGRAND, $\epsilon=0.6$}

\addplot[blue, mark = o, dashed, mark options=solid]
table[]{x y
2 0.01278
2.5 0.012213
3 0.010104
3.5 0.0060259
4 0.0027644
4.5 0.0010183
5 0.00028736
5.5 4.9e-05
};

\addplot[brown, mark = o, dashed, mark options=solid]
table[]{x y
2 0.039817
2.5 0.033974
3 0.024635
3.5 0.01335
4 0.0049254
4.5 0.0017181
5 0.00038143
5.5 6.6e-05
};

\addplot[black, mark = o, dashed, mark options=solid]
table[]{x y
2 0.12074
2.5 0.090216
3 0.051792
3.5 0.023972
4 0.0080455
4.5 0.0022734
5 0.000485
5.5 9e-05

};

\addplot[red, mark = x, dashed, mark options=solid]
table[]{x y
2 0.064346
2.5 0.045863
3 0.027254
3.5 0.014787
4 0.0065436
4.5 0.0024193
5 0.0007509
5.5 0.00018605
};

\end{semilogyaxis}
\end{tikzpicture}
    \vspace{-0.2in}
    \caption{The error detection and correction performance of: 1) a $(64,51+6)$ 5G polar code concatenated with the $6$-bit CRC \texttt{0x30} and with \ac{cascl} decoding; (2) an $\eBCH(64,51)$ code  with \ac{orbgrand} decoding and a threshold-based erasure decision with threshold $\epsilon=0.15$. Solid lines correspond to \ac{bler}, dashed lines to \ac{uer}.}
    \vspace{-0.3in}
    \label{fig:ca-polar-detection}
\end{figure}

\section{Discussion}
We have established that soft input GRAND algorithms can, during their execution, evaluate a predicted likelihood that the decoded block or list is in error. We have derived exact formulae along with readily computable approximations. While the formulae assume random codebooks, we have empirically shown them to make accurate predictions for structured codebooks. 

There are many potential applications of this soft output. It can be used to reduce the rate of undetected errors during decoding or, for URLLC, to do so more cheaply as ORBGRAND can use a \ac{crc} or any other code for both error correction and reduction of undetected errors. In \ac{harq} schemes, the predicted correctness probability could be used to determine whether to request retransmission,  reducing the number of requests. It has been shown \cite{duffy2023_softdetectionphysical} that GRAND soft output can be used to compromise the security of wiretap channels. The confidence measure could also be used to determine the most reliable decoding from a collection of decodings, which could help to select a lead channel in noise recycling \cite{cohen2020_noise,Riaz22}. 

\vspace{-0.1in}
\section{Acknowledgements}
This work was partially supported by Defense Advanced Research Projects Agency contract number HR00112120008. This publication has emanated from research conducted with the financial support of Science Foundation Ireland under grant number 18/CRT/6049. The opinions, findings and conclusions or recommendations expressed in this material are those of the author(s) and do not necessarily reflect the views of the Science Foundation Ireland. 

\bibliographystyle{IEEEtran}
\bibliography{references}

\end{document}